\documentclass[12pt]{article}
\usepackage{graphicx}
\usepackage{amsfonts,amsmath}
\usepackage[mathscr]{eucal}
\usepackage{amssymb}
\usepackage{amsthm}
\usepackage{bbold}
\theoremstyle{plain}
\newtheorem{thm}{Theorem}

\def\theequation{\arabic{section}.\arabic{equation}}
\newcommand{\be}{\begin{eqnarray}}
\newcommand{\ee}{\end{eqnarray}}
\newcommand{\bc}{\begin{center}}
\newcommand{\ec}{\end{center}}
\newcommand{\nn}{\nonumber \\}
\newcommand{\lb}{\label}
\newcommand{\p}[1]{(\ref{#1})}

\begin{document}

\begin{titlepage}

\vspace*{0.2cm}

\renewcommand{\thefootnote}{\star}
\begin{center}

{\LARGE\bf  Comments on noncommutative quantum mechanical systems associated with Lie algebras}

\vspace{2cm}

{\Large Andrei Smilga} \\

\vspace{0.5cm}

{\it SUBATECH, Universit\'e de
Nantes,  4 rue Alfred Kastler, BP 20722, Nantes  44307, France. }

\end{center}
\vspace{0.2cm} \vskip 0.6truecm \nopagebreak

   \begin{abstract}
\noindent  

We consider quantum mechanics on the noncommutative spaces characterized by the commutation relations
$$ [x_a, x_b] \ =\ i\theta f_{abc} x_c\,, $$
where $f_{abc}$ are the structure constants of a Lie algebra. We note that this problem can be reformulated as an ordinary quantum problem in a commuting {\it momentum} space. The coordinates are then represented as linear differential operators $\hat x_a = -i \hat D_a = -iE_{ab} (p)\,
\partial /\partial p_b $. Generically, the matrix $E_{ab}(p)$ represents a certain infinite series over the deformation parameter $\theta$: $E_{ab} = \delta_{ab} + \ldots$. The deformed Hamiltonian,
$\hat H \ =\ - \frac 12  \hat D_a^2\,, $ describes the motion along the corresponding group manifolds with the characteristic size of order $\theta^{-1}$. Their metrics are also expressed into  certain infinite series in $\theta$, with  $E_{ab}$ having the meaning of vielbeins.

For the algebras $su(2)$ and $u(N)$, it has been possible to represent the operators $\hat x_a$ in a simple finite form. A byproduct of our study are new nonstandard
formulas for the metrics on  all the spheres $S^n$, on the corresponding projective spaces $RP^n$ and on $U(2)$. 
   \end{abstract}

\end{titlepage}

\setcounter{footnote}{0}

\setcounter{equation}0

\section{Introduction} 
The quantum mechanical systems defined  on noncommutative spaces   \cite{noncom} have recently attracted a considerable interest \cite{brmog}.\footnote{Also noncommutative {\it field theories} have been intensely studied \cite{Nekrasov}, but they are not the subject of our present discussion.} The simplest such system involves only two dynamic variables $x,y$ with a nontrivial commutator 
 \be
\lb{xy-comm}
[x,y] \ =\ i\theta \,.
\ee
Probably, the best way to deal with such a system \cite{Delduc} is to go over to the  momentum space,
$(p_x, p_y) \equiv (X,Y)$, in which case $X$ and $Y$ commute, while the original noncommuting coordinates $x,y$ can be interpreted as differential operators, which 
 are quantum conterparts of the  velocity components   of the moving
 particle of unit mass. 
\be
x &\equiv& \hat v_X \ =\ -i \frac \partial {\partial X} + \frac \theta 2  Y \,, \nn
y &\equiv& \hat v_Y \ =\ -i \frac \partial {\partial Y} - \frac \theta 2  X \,. \nn
\ee
A naturally defined Hamiltonian, 
$$ H \ =\ \frac 12 (\hat v_X^2 + \hat v_Y^2)\,, $$
describes then the motion on the plane $(X,Y)$ with a homogeneous orthogonal magnetic field $\theta$ of a particle of unit mass. The canonical momenta are $\hat P_X = \hat v_X - \theta Y/2$ and $\hat P_Y = \hat v_Y + \theta X/2$.

\section{Noncommutative systems involving the Lie algebra commutators}
\setcounter{equation}0
The R.H.S. of the commutator \p{xy-comm} is constant. In this case, $\partial_x\,[x,y] = 
\partial_y\,[x,y] = 0$, which is true if the derivative operator satisfies the ordinary Leibnitz rules. More general, for a space of arbitrary dimension with the constant commutators
$[x_a, x_b] = i \theta_{ab}$, one can still consistently postulate 
$$ [\partial_a, x_b] = \delta_{ab}, \qquad [\partial_a, \partial_b] = 0$$
and observe that the Jacobi identities still hold.

But it is not generically true for more complicated noncommutative spaces when $[x_a, x_b]$ exhibit a nontrivial coordinate dependence. In this paper, we consider a particular class of such spaces characterised by the commutators
\be
\lb{xjxk-comm} 
[x_a, x_b] \ =\ i\theta f_{abc}\, x_c\,,
 \ee
where $f_{abc}$ are structure constants of an arbitrary Lie algebra. Such spaces were considered in \cite{Skoda} where it was shown that, to satisfy the Jacobi identities, one has to postulate the following nontrivial commutators:
\be
\lb{Skoda}
[\tilde \partial_a, x_b] \ =\ \left[   \frac{i\theta F(\tilde \partial)}{e^{i\theta F(\tilde \partial)}-1 } \right]_{ab} \ =\ \delta_{ab} + \sum_{n=1}^\infty (-i\theta)^n \frac {B_n^+}{n!} [F^n(\tilde \partial)]_{ab}, \qquad [\tilde \partial_a, \tilde \partial_b] = 0\,,
\ee
where \be
F_{ab}(\tilde\partial) = f_{acb}\, \tilde \partial_c \ee
 and $B_n^+$ are the Bernoulli numbers.\footnote{
The Bernoulli numbers are defined from the Taylor expansion,
$$   \frac t{1 - e^{-t}} \ =\ \sum_{n=0}^\infty \frac {B^+_n t^n}{n!}\,.$$
The first few Bernoulli numbers are $B_0^+ = 1, B_1^+=1/2, B_2^+ = 1/6, B_3^+ = 0$ and  $B_4^+ = -1/30$. Note that $B^+_{2m+1} $ all vanish for $m > 0$.}

 The Jacobi identities for the algebra \p{xjxk-comm}, \p{Skoda} are satisfied, as one can be convinced quite directly in each order of the expansion in $\theta$, using the recursive relations for the Bernoulli numbers. In the limit $\theta \to 0$, the algebra elements $\tilde \partial_a$ acquire the meaning of the ordinary partial derivatives $\partial_a$ with respect to the commuting variables  $x_a$.

In the spirit of our remark in the Introduction, we perform the quantum canonical transformation,
\be
\hat p_a =  -i\tilde  \partial_a  \ &\to& \ -X_a, \nn
x_a  \ &\to& \ \hat P_a,
 \ee
so that the elements $\tilde \partial_a$ are traded 
 for the commuting coordinates $X_a$, whereas the original coordinates $x_a$ are realized as the linear differential operators,
\be
\lb{Dj}
x_a \to  -i\hat D_a\ = \  -i 
\left[  \frac{F(\theta X)}{1 - e^{- F(\theta X)}} \right]_{ab} \frac \partial {\partial X_b} \,.
\ee
The operators $\hat D_a$ satisfy the relations
\be
\lb{alg-D}
[\hat D_a, \hat D_b] \ =\ -\theta f_{abc} \hat D_c \,.
 \ee
In other words, we have constructed a representation of the generators of an arbitrary Lie group not in the matrix form, as is usually done, but in the form of linear differential operators acting on the space of the dimension coinciding with the dimension of the group. The operators $\hat D_a$ have a transparent geometric meaning. They are none other than the generators of  the {\it right} group rotations. Indeed, one can show that
\be
\lb{left-rot}
 e^{i \theta X_a t_a}e^{i\theta \epsilon_a t_a} \ =\ e^{i\theta t_a(X_a + \delta X_a)}\,,
\ee
where\footnote{A simple derivation of this relation (which represents a particular case of the general Baker-Campbell-Hausdorff formula) is given in the appendix.} 
\be 
\lb{delta-X}
\delta X_a \ =\  \sum_{n=0}^\infty \frac {B_n^+ [F^n(\theta X)]_{ab}}{n!} \epsilon_b \ + o(\epsilon) \,.
\ee
And this is, of course, the reason why $\hat D_a$ satisfy the Lie algebra commutation relations.

This paper is devoted to studying the properties of the Hamiltonian 
\be
\label{Hquant}
\hat H \ =\ - \frac 12 \hat D_a^2 \,.
 \ee

The classical Hamiltonian corresponding to the quantum Hamiltonian \p{Hquant} reads
  \be
\label{Hcl}
H^{\rm cl} \ =\ \frac 12 g^{jk} P_j P_k\,,
 \ee
where 
\be
\label{contrmetr}
g^{jk} \ =\ E_a{}^j E_a{}^k
 \ee
with 
\be
\label{R}
E_a{}^j \ =\ \left[  \frac {F(\theta X) }{1 - e^{- F(\theta  X)}} \right]_a^{\ j} \,.
\ee
This Hamiltonian describes the motion over the manifold with the inverse  metric $g^{jk}$. The objects  $E_a{}^j$ have the meaning of the vielbeins.\footnote{Having said that, we went to the middle of the alphabet to denote the world tensor indices, while keeping its beginning for the flat tangent space indices.}
The classical Lagrangian corresponding to the Hamiltonian \p{Hcl} depends on the covariant metric tensor and reads
 \be
\lb{L}
L \ =\ \frac 12 g_{jk} \dot X^j \dot X^k\,.
 \ee

\begin{thm}
The metric $g_{jk}$ in \p{L} coincides with the invariant metric on the group manifold.
\end{thm}
\begin{proof}
The canonical metric on  a group manifold, which is invariant under the left and right group rotations, is\footnote{It is familiar to physicists, appearing e.g. in the effective chiral Lagrangian in QCD \cite{chiral},
$$
{\cal L} \ \propto \ {\rm Tr} \left\{\partial_\mu \omega^{-1} \partial^\mu \omega  \right\}
$$
with $\omega \in SU(2)$ or $SU(3)$.
} 
\be
\lb{Einstein}
g_{jk} \ = \  \frac 1{h\theta^2}  {\rm Tr} \left\{ \partial_j \omega^{-1} \partial_k \omega \right\} \,,
\ee
where
\be
\lb{exp-param}
\omega(X) \ = \ \exp\{i\theta X^j t_j\} 
\ee
and  $t_j$  are the generators in a given  representation with the Dynkin index $h$. Then $g_{jk}(X=0) = \delta_{jk}$.

Consider the distance between two close points $\omega(X + \delta X)$ and $\omega(X)$. For the invariant metric, this distance is the same as the distance between $\omega^{-1}(X)  \omega(X+ \delta X)  = \omega(\epsilon)$ and the group unity. The latter  is $ds^2 = \epsilon_a  \epsilon_a$. To find the metric at the vicinity of $X$, we have only express $\epsilon_a$ in terms of $\delta X^j$. To do so, we use the formula \cite{Feynman} (its derivation is explained in the appendix)
\be
\lb{derexp}
\frac 1\theta \, \omega^{-1}  (-i\partial_j\omega) \ =\ \int_0^1 d\tau\, e^{-\tau  R} t_j e^{\tau R} \  = \ t_j - \frac 12 [R, t_j] + \frac 16 [R, [R, t_j]] 
-  \ldots \ =\ E_j{}^a t_a
 \ee
where $R = i\theta t_j X^j$ and 
\be
\lb{Q}
E_j{}^a \ =\ \left[ \frac {1- e^{-F(\theta X)}}{F(\theta X)} \right]_j^{\ a} \ =\ \sum_{n=0}^\infty \frac {[ F^n(-\theta X)]_j{}^a}{(n+1)!} 
 \ee  
is the inverse vielbein.
Then $\epsilon^a = E_j{}^a \delta X^j$ and the metric is 
\be
\lb{covmetr}
g_{jk} = E_j{}^a E_k{}^a\,,
 \ee
 which matches \p{contrmetr}.
 
\end{proof}

Thus, the classical Hamiltonian \p{Hcl} describes the motion of a particle along the group manifold. 

For $SU(2) \equiv S^3$, the sums for the vielbeins in \p{R} and \p{Q} can be done:\footnote{This formula coincides with Eq. (4.80) in recent \cite{Kupr}, though its interpretation there was completely different.}
 \be
\lb{E-su2}
E_a{}^j \ &=&\ \delta_{aj} + \frac \theta 2 \varepsilon_{apj} X_p  + \left(\delta_{ja} -  \frac{X_a X_j}{X^2} \right)  \left[\frac {\theta X/2}{\tan (\theta X/2)}-1 \right]
\nn 
E_j{}^a \ &=& \ \delta_{ja} -  \frac {2 \sin^2 (\theta X/2)}{\theta X^2} \varepsilon_{jpa}X_p +  \left(\delta_{ja} -  \frac{X_a X_j}{X^2} \right)  \left[ \frac {\sin (\theta X)}{\theta X} -1\right] 
\,,
 \ee
 where $X = \sqrt{X_pX_p}$ and the positions of the indices in the right-hand sides do not have a tensorial meaning and are irrelevant. 

The metric \p{contrmetr}, \p{covmetr} reads
\be
\lb{metr-su2}
g^{jk} \ &=& \ A^{-1}(X) \left(\delta_{jk} - \frac {X_j X_k}{X^2} \right)  + \frac {X_j X_k}{X^2}, \nn
g_{jk}  \ &=&\ A(X) \left(\delta_{jk} - \frac {X_j X_k}{X^2} \right) + \frac {X_j X_k}{X^2}\,,
\ee
where 
$$ A(X) \ =\ \frac {4 \sin^2 (\theta X/2)}{\theta^2 X^2}\,.$$
At the vicinity of the origin, 
\be
g_{jk} \ =\ \delta_{jk} + \frac {\theta^2}{12}(X_j X_k - X^2 \delta_{jk})  + O(X^4)\,.
\ee
The corresponding scalar curvature is $R = 3\theta^2/2$. Bearing in mind that $R = 6/\rho^2$ where $\rho$ is the radius of $S^3$, we derive
\be
\lb{radius}
\rho  \ = \  \frac 2\theta\,.
 \ee
The  metric \p{covmetr}  looks quite different from the familiar conformally flat metric on $S^3$,  
\be
\lb{stereo}
g_{jk} \ =\ \frac {\delta_{jk}}{[1 + Z^2/(4\rho^2)]^2}\,.
\ee
But  these two metrics describe the same geometry, which means that the only difference between them is the  choice of coordinates $Z \to X$ [see Eqs. \p{zamenaYX},\p{zamenaYZ} below].

\vspace{1mm}

Consider now the  quantum Hamiltonian \p{Hquant} implying a particular way of ordering of the coordinates and canonical momenta.

\begin{thm}
The Hamiltonian \p{Hquant} coincides up to the factor $-1/2$ with the Laplace-Beltrami operator on the group manifold,
\be 
\lb{Laplace}
\triangle \ =\  \frac 1{\sqrt{g}} \partial_j (\sqrt{g} g^{jk}) \partial_k\,,
 \ee
with the invariant metric \p{covmetr}.
\end{thm}

\begin{proof}

Clearly the operator $(\hat D_a)^2$ commutes with $\hat D_a$. This means that the Hamiltonian \p{Hquant} is invariant under right group rotations. But it must  also be invariant under left group rotations,
 \be
\lb{right-rot} 
 e^{i \theta t_j X^j} \ \to \  e^{i\theta  t_j \epsilon^j} e^{i \theta t_j X^j}   \,. 
 \ee
Indeed,  right and left  rotations commute and their generators $\hat D_a$ and $\hat D'_a$ must also commute. Both  of them can be represented as infinite sums   
\be
 \hat D_a \ &=& \ \sum_{n=0}^\infty \frac {B_n^+  \,\theta^n [ F^n(X)]_a{}^p}{n!} \, \partial_p\,, \nn
\hat D'_a \ &=& \ \sum_{n=0}^\infty \frac {B_n^- \, \theta^n [ F^n(X)]_a{}^p}{n!} \, \partial_p\,.
\ee
The only difference is that the sum for $\hat D_a$ involves the Bernoulli numbers $B_n^+$, while the sum for  $\hat D'_a$ involves the Bernoulli numbers  
 $B_n^-$. The latter  coincide with $B_n^+$ for all  $n$ except $n=1$ where $B_1^- = - B_1^+ = - 1/2$. 

Left and right rotations commute as well as their generators. The vanishing of the commutator  $[\hat D'_a, \hat D_b] \ =\ 0$ entails the vanishing of $[\hat  D'_a, \hat H^{\rm qu}]$.
The only second order differential operator that is invariant under both left and right group rotations is, up to a constant, the Laplace-Beltrami operator.  
 
\end{proof} 

Note that the operators $\hat D'_a$ satisfy the commutation relations
\be
\lb{alg-D1}
[\hat D'_a, \hat D'_b] \ =\ \theta f_{abc} \hat D'_c
 \ee
with the opposite sign, compared to \p{alg-D}. The origin of this extra sign is quite clear bearing in mind \p{left-rot} and \p{right-rot}. The composition of two right rotations $g_1$ and $g_2$ is $g_1 g_2$, while for left rotations it is $g_2 g_1$.

Another remark is the following. The Hamiltonian \p{Laplace} seems to be not Hermitian.
But it is nothing but an ``optical illusion". It {\it is} Hermitian in the Hilbert space where the
wave functions are normalized with the covariant measure,
 \be
\lb{measure}
\int d^d X \sqrt{g} \left| \Psi(X) \right|^2 \ =\ 1\,,
 \ee
$g = \det \| g_{jk} \|$.
In the Hilbert space with the flat measure $d^d X$, the Hamiltonian would include extra wrapping factors: 
$H^{\rm flat} = g^{1/4} H^{\rm flat} g^{-1/4}$. 

The same concerns the operator $\hat D_a$ defined in \p{Dj}: it is anti-Hermitian in the Hilbert space including the factor $\sqrt{g}$ in the measure.

\vspace{1mm}

We hasten to comment, however, that the representation \p{Dj} for $\hat D_a$ is not unique. For example, for $su(2)$, the following nice representation is possible:\footnote{Cf. Eq. (37) in Ref. \cite{ital}.}
\be
\lb{new-D}
\hat D_a \ =\ \frac \partial {\partial Y_a} - \frac \theta 2 \varepsilon_{abc} Y_b \,\frac \partial {\partial Y_c} + \frac {\theta^2}4 Y_a Y_b \,\frac \partial {\partial Y_b}\,.
\ee
The algebra \p{alg-D} still holds.

The operator $(\hat D_a)^2$ describes then the motion of a particle over the 3-dimensional manifold with the metric 
\be
\lb{metr-S3}
g^{jk} &=& \delta^{jk} + \kappa (Y^2 \delta^{jk} + Y^j Y^k) + \kappa^2 Y^2 Y^j Y^k \,, \nn
g_{jk}  &=& \frac {\delta_{jk}}{1+ \kappa Y^2} - \frac {\kappa Y_j Y_k}{(1 + \kappa Y^2)^2}\,,
 \ee
where $\kappa = \theta^2/4$.

The corresponding Christoffel symbols are
\be
\lb{Gamma}
\Gamma^p_{kl} \ =\ - \frac \kappa {1 + \kappa Y^2} (Y_k \delta_{pk} + Y_l \delta_{pk} )\,.
 \ee
The  Riemann tensor reads 
\be
&&R^i_{kpl} \ =\ \partial_p \Gamma^i_{kl} - \partial_l \Gamma_{kp}^i + \Gamma_{np}^i \Gamma_{kl}^n - \Gamma_{nl}^i \Gamma_{kp}^n \nn
&&=\ \frac {\kappa}{1+ \kappa Y^2}(\delta_{kl} \delta_{ip} -  \delta_{kp} \delta_{il}) + \frac {\kappa^2 Y_k (Y_p \delta_{il} - Y_l \delta_{ip})}{(1 + \kappa Y^2)^2} \ =\ \kappa (\delta^i_p g_{kl} - \delta^i_l g_{kp})\,. 
\ee
 Hence $R_{ikpl} \ =\ \kappa (g_{ip} g_{kl} - g_{il} g_{kp})$. 

It follows that all the curvature invariants are constant:
\be
\lb{invar}
&&R \ = \ 6\kappa = \frac {3\theta^2}2, \qquad  R_{kl} R^{kl} \ =\ 12\kappa^2  \ =\ \frac {3\theta^4}4, \qquad 
R_{kp} R^{pn}R_n^k  \ =\ 24\kappa^3 \ =\  \frac {3\theta^6}8, \nn
 &&R_{ikpl} R^{ikpl} \ =\  12\kappa^2  \ =\ \frac {3\theta^4}4\,, \quad \ldots
\ee
One is then tempted to say that our manifold represents the ``round" 3-sphere of radius \p{radius},
but it is not exactly so.  The determinant of the metric is 
$$ g \ =\ \frac 1 {(1 + \kappa Y^2)^4}.$$
The volume of the manifold,
\be
V \ =\  \iiint_{-\infty}^\infty \sqrt{g} \, d^3 x \ =\ \pi^2 \rho^3\,,
 \ee
is two times smaller than the volume of $S^3$, which means that we are dealing with the projective space  $RP^3 = S^3/Z_2 \equiv SO(3)$.
In fact, one can confirm it, writing an explicit coordinate change \cite{Kuprnew},
that brings the metric \p{metr-S3} in the form \p{metr-su2}:
 \be
\lb{zamenaYX}
Y^p \ =\ X^p \, \frac {2\tan \left(\theta  X/2 \right)} {\theta  X}\,.
 \ee 
 Then $Y^p = \infty$ corresponds to $X = \sqrt{X^p X^p} = \pi/\theta$. And the latter value corresponds to the equator on $S^3$ in the exponential parameterization \p{exp-param}.

Another variable change,
\be
\lb{zamenaYZ}
Y^p \ =\ Z^p \, \frac {1} {1 - \kappa Z^2/4}\,,
 \ee 
brings the metric \p{metr-S3} in the habitual conformally flat form.

\subsection{Higher spheres}

The expressions \p{metr-su2} and  \p{metr-S3} describe the metric of $S^n$ in any dimension. Indeed, the stereographic projection leading to the metric \p{stereo}
can be performed in any dimensions, and the variable changes \p{zamenaYZ}   \p{zamenaYX} bring it in the forms \p{metr-su2}, \p{metr-S3} in any dimension.

However, the manifolds $S^3$ and $S^7$ are special. These spheres are parallelizable and admit quaternion (resp. octonion) algebraic structure. 
The case of $S^3$ we already discussed. It is just a group manifold and the Hamiltonian describing its motion is related to the Lie algebra \p{alg-D}. But the motion over $S^7$ can also be associated with a not so simple nonlinear algebra. 

 Define the operators $\hat D_A  = E_A^J \partial_J$ with the vielbeins 
\be
\lb{EAJ}
E_A{}^J \ = \ \delta_{AJ} + \frac \theta 2 \eta_{APJ} X_P  + \left(\delta_{AJ} -  \frac{X_A X_J}{X^2} \right)  \left[\frac {\theta X/2}{\tan (\theta X/2)}-1 \right]\,.
 \ee
They have the same form as in \p{E-su2}, but $X_J$ are now the coordinates on $S^7$ and the antisymmetric tensor $\eta_{ABC}$ defines the rules of octonion multiplication,\footnote{Different choices for $\eta_{ABC}$ are possible. One of them is
$$ \eta_{123} = \eta_{435} = \eta_{147} = \eta_{165} = \eta_{246} = \eta_{123} = \eta_{367} \ =\ 1\,,$$
the other nonzero components being restored by antisymmetry.}

$e_A e_B = -\delta_{AB} + \eta_{ABC} e_C$. Using the identity 
\be
\lb{et-et}
\eta_{ABC} \eta_{DEC} \ =\ \delta_{AD}\delta_{BE} - \delta_{AE} \delta_{BD} -  \eta_{ABDE}
 \ee
with 
 \be
\lb{eta4}
 \eta_{ABDE} \ =\ \frac 16 \epsilon_{ABDEFGH}\, \eta_{FGH},
 \ee
 one can be convinced that
the convolution of the vielbeins \p{EAJ} gives the metric on $S^7$, which, as we have mentioned, has the same form as in \p{metr-su2}, but with seven coordinates.

The difference is that the operators $\hat D_A$ do not form a Lie algebra as in \p{alg-D}, but their commutators have a more complicated form with the coordinate dependence in the right hand side:\footnote{In different terms, this result was derived in \cite{Kupr-nonlin} using \p{et-et}
and the identity 
$$ \eta_{ABC} \,\eta_{AFDE} \ =\ \delta_{CF}\, \eta_{BDE} +  \delta_{CD}\, \eta_{BEF}  + \delta_{CE}\, \eta_{BFD}   - (B \leftrightarrow C)\,, 
$$
 but we translated  it in our notation. 
} 
 \be
[\hat D_A, \hat D_B] \ &=&\ -\theta \eta_{ABC}\hat D_C +   \frac {\theta \sin (\theta X)} X  \, \eta_{ABCP}X_P \hat D_C  \nn
&+& 2\frac {\theta \sin^2 (\theta X/2)}{X^2} \, \eta_{ABCP} \eta_{CDQ} X_P X_Q \hat D_D\,.
\ee

\section{Gurevich-Saponov model.}
\setcounter{equation}0

\subsection{$u(2)$ model.}

In Refs. \cite{Gurevich1,Gurevich2},  a noncommutative quantum mechanical system was studied with the following nontrivial commutation relations:
\be
\lb{commGur-xt}
 \,[x_a, x_b] \ =\ i\theta \epsilon_{abc} x_c, \qquad[x_a, \tau] \ =\ 0\,,
\ee
supplemented by
 \be
\lb{commGur-der}
&& [\tilde \partial_a, x_b] \ =\ \delta_{ab}\left( 1 + \frac {i\theta}2 \tilde \partial_\tau \right) + \frac {i\theta}2 \epsilon_{abc} \tilde \partial_c, \qquad 
[\tilde \partial_\tau, x_a] \ =\ - \frac {i\theta}2 \tilde \partial_a, \nn 
&&[\tilde \partial_\tau, \tau] \ =\ 1 + \frac {i\theta}2 \tilde \partial_\tau, \qquad 
[\tilde \partial_a, \tau] \ =\ \frac {i\theta}2 \tilde \partial_a\,, \qquad
 [\tilde \partial_a, \tilde \partial_b] \ =\ [\tilde \partial_a, \tilde \partial_\tau] \ =\ 0\,.
 \ee
It is easy to check that the Jacobi identities hold. This model belongs to the large class of the models \p{xjxk-comm}, but the algebra here is $u(2)$, which is not simple. Note that the R.H.S. of the commutators in \p{commGur-der} does not represent an infinite series as in \p{Skoda}, but has a simple form involving only the first generalized derivatives [cf.  Eq. \p{new-D}, which, in contrast to \p{Dj}, involves only the linear and quadratic terms in $\tilde \partial_a \to -iY_a$ ].

Proceeding in the same way as before, we may trade the coordinates for momenta and vice versa to arrive at the ordinary QM system with four commuting dynamic variables $X_a, T$ and former coordinates and now momenta being realized as the differential operators\footnote{We added a constant in $P_0$ to make it Hermitian with the flat measure $d\mu = d^3X DT$.}
\be
\lb{mom-Gur}
P_a \ =\ -i \hat D_a \ &=&\ -i \left[ \left( 1 + \frac {\theta T}2\right) \frac \partial {\partial X_a} - \frac \theta 2 X_a \frac \partial {\partial T} + \frac \theta 2 \epsilon_{abc} X_b \frac \partial {\partial X_c} \right]\,, \nn 
P_0\ =\ -i\hat D_0 \ &=& \ -i \left[\left( 1 + \frac {\theta T}2 \right) \frac {\partial}{\partial T} + \frac \theta 2 X_a \frac \partial {\partial X_a} + \theta \right]\,.
 \ee

The algebra $$[\hat D_a, \hat D_b] = - \theta \epsilon_{abc} \hat D_c, \qquad [\hat D_a, \hat D_0] = 0$$ holds. 

Consider now the operator 
\be
\lb{HGur}
\hat H \ =\ -\frac 12 (\hat D_0 \hat D_0 + \hat D_a \hat D_a)\,.
 \ee
Its classical counterpart is
\be
\lb{Hcl-Gur} 
H^{\rm cl} \ =\ \frac {P_0^2 + P_a^2}2 \left[1 + \theta T + \frac {\theta^2}4 (T^2 + X_a^2) \right]\,.
 \ee
This Hamiltonian describes the motion of a particle over a 4-dimensional manifold with the conformally flat metric 
\be
\lb{metrGur}
g_{\mu\nu}  = \ \frac {\delta_{\mu\nu}}{ F(T, X_a)} \ =\ \frac {\delta_{\mu\nu}} {1 + \theta T + \theta^2 (T^2 + X_a^2)/4}\,.
\ee
 The Ricci tensor has the following components:
 \be
\lb{Ricci}
R_{00} \ =\ \frac {\theta^4}{8F^2}, \qquad R_{0a} \ =\ - \frac {\theta^3 (1 + \theta T/2) X_a}{4F^2}, \qquad R_{ab } \
 =\ \frac {\theta^2}{2F} \left( \delta_{ab} - \frac {\theta^2 X_a X_b}{4F} \right)\,.
 \ee
An explicit calculation shows that all the curvature invariants made from the Ricci tensor are constants:
\be
\lb{Ricciinv}
R = R_{\mu\nu} g^{\mu\nu} \ =\ \frac {3\theta^2}2\,, \qquad R_{\mu\nu} R^{\mu\nu} \ =\  \frac {3\theta^4}4, \qquad R_{\mu\nu} R^{\nu\rho}
R_\rho{}^\mu \ =\ \frac {3\theta^6}8, \ \ {\rm etc.}
 \ee
They are the same as for a {\it 3-dimensional} ``round" sphere of radius \p{radius} [see Eq. \p{invar}].

 In other words, our manifold is\footnote{It is not $S^3 \times S^1$ or $SO(3) \times S^1$  because the integral 
$$\int d^3 X dT \sqrt{g} =  \int \frac {d^3 X dT}{F^2}$$
diverges and the manifold is not compact.} $S^3 \times \mathbb{R}$ or maybe $SO(3) \times \mathbb{R}$ ---  it is difficult to say based solely on the form
\p{metrGur} of the metric where the coordinates $X_a$ and $T$ are intertwined. 
This result is quite natural as the Gurevich-Saponov model involves the nontrivial commutators \p{commGur-xt} of the $u(2)$ algebra. After interchanging the coordinates and momenta, we are arriving at the system that describes the motion along  $U(2)$ in a non-compact realization.   

\subsection{Higher $N$}

One can also consider \cite{Gurevich2} the $u(N)$ generalization of the algebra  \p{commGur-xt}, \p{commGur-der}. Introduce $N^2 + N^2$ elements 
$l_j^n$, $\tilde \partial_j^n$ and postulate the following commutators:
 \be
\lb{alg-uN}
\, [l_j^n, \l_k^m] \ =\ \theta (\delta_j^m l_k^n - \delta_k^n l_j^m),  \qquad
 [\tilde \partial_j^n, l_k^m] \ = \ \delta_j^m \delta_k^n + \theta \delta_j^m \tilde \partial_k^n\,. 
 \ee
The first relation is the habitual $u(N)$ algebra. It is supplemented by the commutation laws for the quasiderivatives $\tilde \partial_i^n$. It is not difficult to check that the Jacobi identities are satisfied. 

To make contact with \p{commGur-xt}, \p{commGur-der}, we introduce the notation
\be
\lb{compact-var}
x_A \ =\ (t_A)_n^j \, l_j^n, \qquad \tilde \partial_A = 2(t_A)_n^j \, \tilde \partial_j^n\,,
 \ee
where $A = (a,0)$ and $t_A$ are the generators of $U(N)$ in the fundamental representation --- the Hermitian $N \times N$ matrices   normalized in a habitual way
$${\rm Tr} \{t_A t_B\} \ = \ \frac12 \delta_{AB}\,.$$
In particular, 
$ t_0  = \mathbb{1}/\sqrt{2N}$, so that $x_0 \equiv  t = l_j^j/\sqrt{2N}$ and $\tilde \partial_t = \sqrt{2/N} \, \tilde \partial_j^j$. It is also convenient to introduce $\tau = it, \tilde \partial_\tau = -i \tilde \partial_t$. In these terms, the algebra \p{alg-uN} reads   

 \be
\lb{alg-uN-comp}
&&[x_a, x_b] \ =\ i\theta f_{abc} x_c, \qquad [x_a, \tau] \ =\ 0, \nn
&&[\tilde \partial_a, x_b] \ =\ \delta_{ab} \left( 1 + \frac {i\theta}{\sqrt{2N}} \,\tilde \partial_\tau \right) + \frac {i\theta}2 (f_{abc} - i d_{abc}) \, \tilde \partial_c, \nn
&& [\tilde \partial_a, \tau] \ =\  - [\tilde \partial_\tau, x_a]  \ =\  \frac {i\theta}{\sqrt{2N}} \tilde \partial_a, \qquad [\tilde \partial_\tau, \tau] \ =\ 
1 + \frac {i\theta}{\sqrt{2N}} \, \tilde \partial_\tau\,,
 \ee
where $f_{abc}$ are the $su(N)$ structure constants and $d_{abc}$ is a symmetric tensor entering the relation
$$ {\rm Tr}\{t_a t_b t_c \} \ =\ \frac 14 (i f_{abc} + d_{abc})\,.$$
For $N=2$, $d_{abc}$ vanishes and \p{alg-uN-comp} coincides with \p{commGur-xt}, \p{commGur-der}.

We perform now our quantum canonical transformation to represent this algebra in terms of the commuting coordinates (former momenta) $X_a, T$ and the differential operators
\be
\lb{mom-Gur}
x_a \to P_a \ =\ -i \hat D_a \ &=&\ -i \left[ \left( 1 + \frac {\theta T}{\sqrt{2N}}\right) \frac \partial {\partial X_a} - \frac \theta {\sqrt{2N}} X_a \frac \partial {\partial T} + \frac \theta 2 (f_{abc} + i d_{abc}) X_b \frac \partial {\partial X_c} \right]\,, \nn 
\tau \to P_0\ =\ -i\hat D_0 \ &=& \ -i \left[\left( 1 + \frac {\theta T}{\sqrt{2N}} \right) \frac {\partial}{\partial T} + \frac \theta {\sqrt{2N}} X_a \frac \partial {\partial X_a} \right]\,.
 \ee
 The algebra $$[\hat D_a, \hat D_b] = - \theta f_{abc} \hat D_c, \qquad [\hat D_a, \hat D_0] = 0$$ holds.\footnote{It can be checked explicitly using the identities
(see e.g. Appendix A of Ref. \cite{Fadin})
\be
\lb{identities}
&&f_{abe} f_{cde}  + f_{cae} f_{bde} + f_{bce} f_{ade}  \ =\ 0, \nn
&&f_{abe} d_{cde} + f_{ace}  d_{dbe} +  f_{ade}  d_{bce} \ =\ 0, \nn
&&f_{abe} f_{cde} \ =\ \frac 2N (\delta_{ac} \delta_{bd} - \delta_{bc} \delta_{ad}) + d_{ace} d_{bde} -  d_{bce} d_{ade}\,.
\ee
}
Note now the intrinsic complexity $\sim f_{abc} + i d_{abc}$ in the expression for $\hat D_a$. That means that the would-be vielbein
\be
E^j_a \ =\ \left( 1 + \frac {\theta T}{\sqrt{2N}} \right) \delta_{aj} + \frac \theta 2 (f_{abj} + i d_{abj}) X_b 
\ee
is complex, and the same concerns the metric $g^{jk} = E^j_A E^k_A$. The determinant of the latter is also complex.

Complex metric and complex volume have little geometric sense. Also the canonical momenta $\hat D_a$ and the Hamiltonian \p{HGur}  are not Hermitian for $N >2$. There are  some Hamiltonians that do not seem to be  Hermitian, but are in fact Hermitian in disguise and have a real spectrum \cite{Bender}. We do not think that it is so in our case, but a special study of this question would be interesting.

\section*{Acknowledgements}
I am deeply indebted to Dmitry Gurevich for many illuminating discussions and the collaboration at the initial stage of this work. I thank Vladislav Kupriyanov for important comments, for bringing my attention to Refs. \cite{Kupr,ital,Kupr-nonlin} and for informing me about the paper \cite{Kuprnew} prior to publication. I also acknowledge
 useful discussions with Maxim Kontsevich, Zoran \v{S}koda and Arkady Vainshtein.

\section*{Appendix}
\def\theequation{A.\arabic{equation}}
\setcounter{equation}0
To make the paper self-contained, we will explain here how the identities \p{derexp} and \p{delta-X} are derived.

The formula \p{derexp} has a nice physical interpretation \cite{Feynman}. Consider the Hamiltonian $\hat H = \hat H_0 + \hat V$ and the corresponding evolution operator
\be 
U \ =\ e^{i(\hat H_0 + \hat V)t} \ =\ \lim_{N \to \infty} \left[ 1 + \frac {it}N (\hat H_0 + \hat V) \right]^N\,.
\ee
Suppose that the perturbation $\hat V$ is small and keep only the linear in $\hat V$ terms. The perturbation can be inserted at any time moment $t'$ between 0 and $t$, and we obtain
\be
U \ =\ e^{i\hat H_0 t} + i\int_0^t dt' e^{i(t - t') \hat H_0} \hat V  e^{it' \hat H_0} + \ldots \,. 
 \ee
Choosing $\hat H_0 = \theta X^j t_j, \ \hat V = \theta \,\delta X^j t_j, \ t=1$, renaming $t' \to \tau$ and recalling the notation $\omega = e^R = e^{i\theta X^j t_j}$, we derive
\be
\omega(X+ \delta X) = \omega(X) + \delta X^j \partial_j \omega(X) + \ldots\ = \
\omega(X)\left[ 1 + i \theta \int_0^1 d\tau\, e^{-\tau R} \delta X^j t_j e^{\tau R} + \ldots\right],
  \ee
from which \p{derexp} follows.

 As a consequence, we derive for the product in \p{left-rot} that $\epsilon^a$ is equal to $E^a_j \,\delta X^j$, with the vielbein $E^a_j$ having the form \p{Q}. Then $\delta X^j$ is expressed via $\epsilon^a$ with the vielbein $E_a^j$ written in \p{delta-X}.

\end{document}